\documentclass{article}
\usepackage{amssymb}
\usepackage{amsmath}
\usepackage{tikz}
\usepackage{ifthen}
\usepackage{listings}
\usepackage{amsthm}
\usepackage{lmodern}
\usepackage{tabularx}
\usepackage{paralist}
\usepackage{booktabs}
\usepackage{subfigure}

\usetikzlibrary{shapes.geometric}
\usetikzlibrary{decorations.pathreplacing}
\usetikzlibrary{decorations.text}
\usetikzlibrary{fit}
\usetikzlibrary{decorations.pathmorphing}
\usetikzlibrary{shapes.multipart}
\usetikzlibrary{shapes.callouts}
\usetikzlibrary{positioning}
\usetikzlibrary{arrows}
\newcommand{\knkcut}{(k,n-k)-cut}
\theoremstyle{remark}     \newtheorem{example}{Example}
\theoremstyle{definition} \newtheorem{algo}[example]{Algorithm}
\theoremstyle{definition} \newtheorem{theorem}[example]{Theorem}
\theoremstyle{definition} \newtheorem{lemma}[example]{Lemma}
\theoremstyle{definition} \newtheorem{corollary}[example]{Corollary}
\theoremstyle{definition} 
\theoremstyle{definition} \newtheorem{definition}[example]{Definition}
\theoremstyle{definition} 
\theoremstyle{definition} \newtheorem{proposition}[example]{Proposition}

\def \goal{\mathop{\mathrm{goal}}}
\def \val{\mathop{\mathrm{val}}}

\let\leq\leqslant
\let\geq\geqslant

\makeatletter
\def\@fnsymbol#1{\ensuremath{\ifcase#1\or *\or (a)\or (b)\or (c)\or (d)\or \S \or
   \mathsection\or \mathparagraph\or \|\or **\or \dagger \or \ddagger \or \dagger\dagger
   \or \ddagger\ddagger \else\@ctrerr\fi}}
\makeatother

\lstdefinelanguage{algo}{%
  morekeywords={input,output,initialise,such,that,variable,variables,while,do,end,
                If,endIf,For,EndFor,endFor,each,While,endWhile,then,Then,Else,for,
                from,to,return,else,pick,choose,set,Compute,Determine,compute,True,False},
  sensitive=true,
  mathescape=true,
}
\title{Multi-parameter complexity analysis for constrained size graph problems: using greediness for parameterization\footnote{Research supported by the French Agency for Research under the program TODO, ANR-09-EMER-010}}

\begin{document}

\author{\'E. Bonnet \hspace*{2mm} B. Escoffier \hspace*{2mm} V. Th. Paschos\footnote{Institut Universitaire de France} \hspace*{4.5mm}  \'E. Tourniaire
\vspace*{2mm} \\
PSL Research University, Universit\'e Paris-Dauphine, LAMSADE \\
CNRS, UMR~7243, France \\
\texttt{\{bonnet,escoffier,paschos,tourniaire\}@lamsade.dauphine.fr}}


\maketitle

\begin{abstract}
We study the parameterized complexity of a broad class of problems
called ``local graph partitioning problems'' that includes the classical fixed cardinality problems as \textsc{max $k$-vertex cover}, \textsc{$k$-densest subgraph}, etc. By developing a
technique ``greediness-for-parameterization'', we obtain fixed
parameter algorithms with respect to a pair of parameters~$k$, the
size of the solution (but \textit{not} its value) and~$\Delta$, the
maximum degree of the input graph. In particular, greediness-for-parameterization improves asymptotic running times for these problems upon random separation (that is a special case of color coding) and is more intuitive and simple. Then, we show how
these results can be easily extended for getting
standard-parameterization results (i.e., with parameter the value of
the optimal solution) for a well known local graph partitioning
problem.
\end{abstract}

\section{Introduction}\label{intro}

A local graph partitioning problem is a problem defined on some
graph $G=(V,E)$ with two integers~$k$ and~$p$. Feasible solutions
are subsets  $V' \subseteq V$ of size exactly~$k$. The value of
their solutions is a linear combination of sizes of edge-subsets and
the objective is to determine whether there exists a solution of
value at least or at most~$p$. Problems as \textsc{max $k$-vertex
cover}, \textsc{$k$-densest subgraph}, \textsc{$k$-lightest
subgraph}, \textsc{max $(k,n-k)$-cut} and \textsc{min
$(k,n-k)$-cut}, also known as fixed cardinality problems, are local
graph partitioning problems. When dealing with graph problems,
several natural parameters, other than the size~$p$ of the optimum,
can be of interest, for instance, the maximum degree~$\Delta$ of the
input graph, its treewidth, etc. To these parameters, common for any
graph problem, in the case of local graph partitioning problem
handled here, one more natural parameter of very great interest can
be additionally considered, the size~$k$ of~$V'$. For instance, the
most of these problems have mainly been studied
in~\cite{cai,Downey03cuttingup}, from a parameterized point of view,
with respect to parameter~$k$, and have been proved W[1]-hard.
Dealing with standard parameterization, the only problems that, to
the best of our knowledge, have not been studied yet, are the
\textsc{max $(k,n-k)$-cut} and the \textsc{min $(k,n-k)$-cut}
problems.

In this paper we develop a technique for obtaining
multi-parameterized results for local graph partitioning problems.
Informally, the basic idea behind it is the following. Perform a
branching with respect to a vertex chosen upon some greedy
criterion. For instance, this criterion could be to consider some
vertex~$v$ that maximizes the number of edges added to the solution
under construction. Without branching, such a greedy criterion is
not optimal. However, if at each step either the greedily chosen
vertex~$v$, or some of its neighbors (more precisely, a vertex at
bounded distance from $v$) are a good choice (they are in an optimal
solution), then a branching rule on neighbors of~$v$ leads to a
branching tree whose size is bounded by a function of~$k$
and~$\Delta$, and at least one leaf of which is an optimal solution.
This method, called ``greediness-for-parameterization'', is
presented in Section~\ref{greedy-param} together with interesting
corollaries about particular local graph partitioning problems.

The results of Section~\ref{greedy-param} can sometimes be easily extended to
standard parameterization results. In Section~\ref{standardparam} we
study standard parameterization of the two still unstudied fixed
cardinality problems \textsc{max} and \textsc{min $(k,n-k)$-cut}. We
prove that the former is fixed parameter tractable~(FPT), while,
unfortunately, the status of the latter one remains still unclear.
In order to handle \textsc{max $(k,n-k)$-cut} we first show that
when $p \leqslant k$ or $p \leqslant \Delta$, the problem is
polynomial. So, the only ``non-trivial'' case occurs when $p > k$
and $p > \Delta$, case handled by greediness-for-parameterization.
Unfortunately, this method concludes inclusion of \textsc{min
$(k,n-k)$-cut} in FPT only for some particular cases. Note that in a
very recent technical report by~\cite{fomincorr}, Fomin et al., the following problem is considered: given a graph $G$ and two integers $k,p$, determine whether there exists a set $V'\subset V$ of size {\it at
most}~$k$ such that at most~$p$ edges have exactly one endpoint in~$V'$. They prove that this problem is FPT with respect to $p$. Let us underline the fact that looking for a set of size at most
$k$ seems to be radically different that looking for a set of size
exactly $k$ (as in \textsc{min $(k,n-k)$-cut}). For instance, in the
case $k=n/2$, the former becomes the \textsc{min
cut} problem that is polynomial, while the latter becomes the \textsc{min bisection} problem that is NP-hard..

In Section~\ref{paramapprox}, we mainly revisit the parameterization
by~$k$ but we handle it from an approximation point of view. Given a
problem~$\Pi$ parameterized by parameter~$\ell$ and an instance~$I$
of~$\Pi$, a parameterized approximation algorithm with ratio $g(.)$
for~$\Pi$ is an algorithm running in time $f(\ell)|I|^{O(1)}$ that
either finds an approximate solution of value at least/at
most~$g(\ell)\ell$, or reports that there is no solution of value at
least/at most~$\ell$. We prove that, although W[1]-hard for the
exact computation, \textsc{max $(k,n-k)$-cut} has a parameterized
approximation schema with respect to~$k$ and \textsc{min
$(k,n-k)$-cut} a randomized parameterized approximation schema.
These results exhibit two problems which are hard with respect to a
given parameter but which become easier when we relax exact
computation requirements and seek only (good) approximations. To our
knowledge, the only other problem having similar behaviour is
another fixed cardinality problem, the \textsc{max $k$-vertex cover}
problem, where one has to find the subset of~$k$ vertices which
cover the greatest number of edges~\cite{marx-approx}. Note that the
existence of problems having this behaviour but with respect to the
standard parameter is an open (presumably very difficult to answer)
question in~\cite{marx-approx}.  Let us note that polynomial
approximation of \textsc{min $(k,n-k)$-cut} has been studied
in~\cite{Feigeminkcut} where it is proved that, if $k=O(\log n)$,
then the problem admits a randomized polynomial time approximation
schema, while, if $k = \Omega(\log n)$, then it admits an
approximation ratio $(1+\frac{\varepsilon k}{\log n})$, for any
$\varepsilon > 0$. Approximation of \textsc{max $(k,n-k)$-cut} has
been studied in several papers and a ratio 1/2 is achieved
in~\cite{ageev} (slightly improved with a randomized algorithm
in~\cite{feige}), for all~$k$.

Finally, in Section~\ref{otherparam}, we handle parameterization of local graph partitioning problems by the treewidth~$\mathrm{tw}$  of the input graph and show, using a standard dynamic programming technique, that they admit an $O^*(2^{\mathrm{tw}})$-time FPT algorithm, when the~$O^*(\cdot)$ notation ignores polynomial factors. Let us note that the interest of this result, except its structural
aspect (many problems for the price of a single algorithm), lies also in the fact that some local partitioning problems (this is the case, for instance, of \textsc{max} and \textsc{min $(k,n-k)$-cut}) do not fit Courcelle's Theorem~\cite{courcellemso2}. Indeed, \textsc{max} and \textsc{min bisection} are not expressible in MSO since the equality of the cardinality of two sets is not MSO-definable. In fact, if one could express that two sets have the same cardinality in MSO, one would be able to express in~MSO the fact that a word has the same number of a's and b's, on a two-letter alphabet, which would make that the set $E=\{w: |w|_a=|w|_b\}$ is MSO-definable. But we know that, on words, MSO-definability is equivalent to recognizability; we also know by the standard pumping lemma (see, for instance,~\cite{lp}) that~$E$ is not recognizable~\cite{maneth}, a contradiction. Henceforth, \textsc{max} and \textsc{min $(k,n-k)$-cut} are not expressible in~MSO; consequently, the fact that those two problems, parameterized by~tw are FPT cannot be obtained by
Courcelle's Theorem. Furthermore, even several known extended variants of~MSO which capture more problems~\cite{szeider}, does not
seem to be able to express the equality of two sets either.

\section{Greediness-for-parameterization}\label{greedy-param}

We first formally define the class of local graph paritioning problems.
\begin{definition}\label{local}
A local graph partitioning problem is a problem having as input a
graph $G=(V,E)$ and two integers~$k$ and~$p$. Feasible solutions are
subsets  $V' \subseteq V$ of size exactly~$k$. The value of a solution,
denoted by~$\mathrm{val}(V')$, is a linear combination $\alpha_1m_1+\alpha_2m_2$ where
$m_1=|E(V')|$, $m_2=|E(V',V\setminus V')|$ and $\alpha_1, \alpha_2 \in \mathbb{R}$.
The goal is to determine whether there exists a solution of value at least $p$
(for a maximization problem) or at most $p$ (for a minimization problem).
\end{definition}
Note that $\alpha_1=1$, $\alpha_2=0$ corresponds to
\textsc{$k$-densest subgraph} and \textsc{$k$-sparsest subgraph},
while $\alpha_1=0$, $\alpha_2=1$ corresponds to
\textsc{$(k,n-k)$-cut}, and $\alpha_1 = \alpha_2=1$ gives
\textsc{$k$-coverage}.
As a local graph
partitioning problem is entirely defined by $\alpha_1$, $\alpha_2$
and $\goal{} \in \{\min,\max\}$ we will unambiguously denote by
$\mathcal L(\goal{},\alpha_1,\alpha_2)$ the corresponding problem.
For conciseness and when no confusion is possible, we will use
\textit{local problem} instead. In the sequel,~$k$ always denotes
the size of feasible subset of vertices and~$p$ the standard
parameter, i.e., the solution-size. Moreover, as a partition into
$k$ and $n-k$ vertices, respectively, is completely defined by the
subset~$V'$ of  size~$k$, we will consider it to be the solution. A
\textit{partial solution}~$T$ is a subset of~$V'$ with less than~$k$
vertices. Similarly to the value of a solution, we define the value
of a partial solution, and denote it by $\val{}(T)$.

Informally, we devise incremental algorithms for local problems that add vertices to an initially empty set~$T$ (for
``taken'' vertices) and stop when~$T$ becomes of size~$k$, i.e.,  when~$T$ itself becomes a feasible solution.
A vertex introduced in~$T$ is irrevocably introduced there and will be not removed later.
\begin{definition}\label{contribution}
Given a local graph partitioning problem $\mathcal L(\goal{},\alpha_1,\alpha_2)$,
the \textit{contribution} of a vertex $v$ within a partial solution $T$
(such that $v \in T$) is defined by
$\delta(v,T) = \frac{1}{2}\alpha_1
|E(\{v\},T)|+\alpha_2|E(\{v\},V \setminus T)|$
\end{definition}
Note that the value of any (partial) solution $T$ satifies
$\val{}(T)=\Sigma_{v \in T} \delta(v,T)$. One can also remark that
$\delta(v,T)=\delta(v,T \cap N(v))$, where~$N(v)$ denotes the (open) neighbourhood of the vertex~$v$.
Function~$\delta$ is called the \textit{contribution function} or simply
the \textit{contribution} of the corresponding local problem.
\begin{definition}
Given a local graph partitioning problem $\mathcal L(\goal{},\alpha_1,\alpha_2)$, a contribution function is said to be \textit{degrading} if for every~$v$, $T$ and~$T'$ such that $v \in T \subseteq T'$, $\delta(v,T) \leqslant \delta(v,T')$ for $\goal{}=\min$
(resp., $\delta(v,T) \geqslant \delta(v,T')$ for $\goal{}=\max$).
\end{definition}
Note that it can be easily shown that for a maximization problem, a
contribution function is degrading if and only if $\alpha_2\geq
\alpha_1/2$ ($\alpha_2\leq \alpha_1/2$ for a minimization problem).
So in particular \textsc{max $k$-vertex cover}, \textsc{$k$-sparsest
subgraph} and \textsc{max $(k,n-k)$-cut} have a degrading
contribution function.
\begin{theorem}\label{th1}
Every local partitioning problem having a degrading
contribution function can be solved in $O^*(\Delta^k)$.
\end{theorem}
\begin{proof}
With no loss of generality, we carry out the proof for a minimization
local problem $\mathcal L(\min,\alpha_1,\alpha_2)$. We recall that $T$
will be a partial solution and eventually a feasible solution.
Consider the following algorithm \texttt{ALG1} which branches
upon the closed neighborhood $N[v]$ of a vertex $v$ minimizing the
greedy criterion $\delta(v,T \cup \{v\})$.
\begin{algo}[\texttt{ALG1}($T$,$k$)] Set $T = \emptyset$;
\begin{itemize}
\item if $k > 0$ then:
\begin{itemize}
\item pick the vertex $v \in V \setminus T$ minimizing $\delta(v,T \cup \{v\})$;
\item for each vertex $w \in N[v] \setminus T$ run $\texttt{ALG1}(T \cup \{w\}$,$k-1)$;
\end{itemize}
\item else ($k=0$), store the feasible solution $T$;
\item output the best among the solutions stored.
\end{itemize}
\end{algo}
The branching tree of \texttt{ALG1} has depth $k$, since
we add one vertex at each recursive call, and arity at most
$\max_{v \in V}|N[v]|=\Delta+1$, where~$N[v]$ denotes the closed neighbourhood of~$v$. Thus, the algorithm runs in~$O^*(\Delta^k)$.

For the optimality proof, we use a classical hybridation technique between
some optimal solution and the one solution computed by \texttt{ALG1}.

Consider an optimal solution $V'_{\mathrm{opt}}$ different from the solution~$V'$ computed by \texttt{ALG1}.
A node $s$ of the branching tree has two characteristics: the
partial solution $T(s)$ at this node (denoted simply $T$ if no
ambiguity occurs) and the vertex chosen by the greedy criterion
$v(s)$ (or simply $v$).
We say that a node~$s$ of the branching tree is \textit{conform}
to the optimal solution~$V'_{\mathrm{opt}}$ if $T(s) \subseteq V'_{\mathrm{opt}}$.
A node $s$ \textit{deviates} from the optimal solution $V'_{\mathrm{opt}}$
if none of its sons is conform to~$V'_{\mathrm{opt}}$.

We start from the root of the branching tree and, while possible, we
move to a conform son of the current node. At some point we reach a node $s$
which deviates from~$V'_{\mathrm{opt}}$. We set $T=T(s)$
and $v=v(s)$. Intuitively, $T$ corresponds to the shared
choices between the optimal solution and \texttt{ALG1}  made along the branch from the root to the node $s$
of the branching tree. Setting $V_n = V'_{\mathrm{opt}} \setminus T$,~$V_n$ does not intersect $N[v]$, otherwise $s$ would
not be deviating.

\begin{figure}
\label{fig1}
\begin{center}
%
%
%
%
%
%
%
\includegraphics[scale=.8]{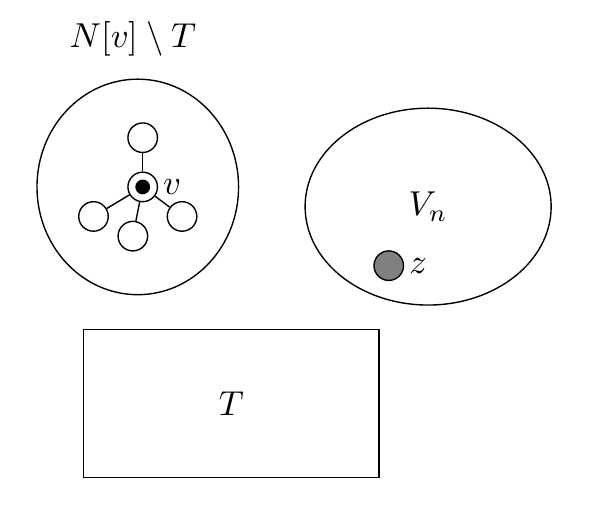}
\caption{Situation of the input graph at a deviating node of the branching tree. The vertex $v$ can substitute $z$ since,
by the hypothesis, $N[v] \setminus T$ and $V_n$ are disjoint and the contribution of a vertex can only decrease when we later
add some of its neighbors in the solution.}
\end{center}
\end{figure}

Choose any $z \in V'_{\mathrm{opt}} \setminus T$ and consider the
solution induced by the set $V_e=V'_{\mathrm{opt}}\cup \{v\}
\setminus \{z\}$. We show that this solution is also optimal. Let
$V_c = V'_{\mathrm{opt}} \setminus \{z\}$. We have
$\val{}(V_e)=\Sigma_{w \in V_c}\delta(w,V_e)+\delta(v,V_e)$.
Besides, $\delta(v,V_e) = \delta(v,V_e \cap N(v)) = \delta(v,T \cup
\{v\})$ since $V_e \setminus (T \cup \{v\})=V_n$ and according to
the last remark of the previous paragraph, $N(v) \cap V_n =
\emptyset$. By the choice of $v$, $\delta(v,T \cup \{v\}) \leqslant
\delta(z,T \cup \{z\})$, and, since $\delta$ is a degrading
contribution, $\delta(z,T \cup \{z\}) \leqslant
\delta(z,V'_{\mathrm{opt}})$. Summing up, we get $\delta(v,V_e)
\leqslant \delta(z,V'_{\mathrm{opt}})$ and $\val{}(V_e) \leqslant
\Sigma_{w \in V_c}\delta(w,V_e) +\delta(z,V'_{\mathrm{opt}})$. Since
$v$ is not in the neighborhood of $V'_{\mathrm{opt}} \setminus T =
V_n$ only $z$ can degrade the contribution of those vertices, so
$\Sigma_{w \in V_c} \delta(w,V_e) \leqslant \Sigma_{w \in
V_c}\delta(w,V'_{\mathrm{opt}})$, and $\val{}(V_e) \leqslant
\Sigma_{w \in V_c}\delta(w,V'_{\mathrm{opt}})
+\delta(z,V'_{\mathrm{opt}}) = \val{}(V'_{\mathrm{opt}})$.

Thus, by repeating this argument at most $k$
times, we can conclude that the solution computed
by \texttt{ALG1} is as good as~$V'_{\mathrm{opt}}$.
\end{proof}
\begin{corollary}\label{cor6}
 \textsc{max $k$-vertex cover}, \textsc{$k$-sparsest subgraph} and \textsc{max $(k,n-k)$-cut}
 can be solved in $O^*(\Delta^k)$.
\end{corollary}
As mentioned before, the local problems mentioned in Corollary~\ref{cor6}
have a degrading contribution.
\begin{theorem}\label{th2}
Every local partitioning problem can be solved in $O^*((\Delta k)^{2k})$.
\end{theorem}
\begin{proof}
Once again, with no loss of generality, we prove the theorem in the case of minimization, i.e., $\mathcal L(\min,\alpha_1,\alpha_2)$.
The proof of Theorem~\ref{th2} involves an algorithm fairly similar
to \texttt{ALG1} but instead of branching on a vertex chosen
greedily and its neighborhood, we will branch on sets of vertices
inducing connected components (also chosen greedily) and the
neighborhood of those sets.

Let us first state the following straightforward lemma that bounds  the number of induced connected components and the running
time to enumerate them.
\begin{lemma}\label{connectedinduced}
One can enumerate the connected induced subgraphs of size up to~$k$
in time $O^*(\Delta^{2k})$.
\end{lemma}
\begin{proof}[Proof of Lemma~\ref{connectedinduced}]
 One can easily enumerate with no redundancy all the
connected induced subgraph of size $k$ which contains a vertex $v$.
Indeed, one can label the vertices of a graph $G$ with integers from
$1$ to $n$, and at each step, take the vertex in the built connected
component with the smaller label and decide once and for all which
of its neighbors will be in the component too. That way, you get
each connected induced component in a unique manner.

Now, it boils down to counting the number of connected induced
subgraph of size~$k$ which contains a given vertex $v$. We denote
that set of components by~$\mathcal C_{k,v}$. Let us show that there
is an injection from~$\mathcal C_{k,v}$ to the set~$\mathcal
B_{k\lceil \log \Delta \rceil}$ of the binary trees with~$k\lceil
\log \Delta \rceil$ nodes.

Recall that the vertices of $G$ are labeled from $1$ to $n$. Given a
component $C \in \mathcal C_{k,v}$, build the following binary tree.
Start from the vertex $v$. From the complete binary tree of height
$\lceil \log \Delta \rceil$, owning a little more than $\Delta$
ordered leaves, place in those leaves the vertices of $N(v)$
according to the order $\leqslant$, and keep only the branches
leading to vertices in $C \cap N(v)$. Iterate this process until you
get all the vertices of $C$ exactly once. When a vertex of $C$
reappears, do not keep the corresponding branch. That way, you get
for each vertex of $C$ a branch of size $\lceil \log \Delta \rceil$,
and hence there are $k \lceil \log \Delta \rceil$ nodes in the tree.

Recall that $|\mathcal B_{k\lceil \log \Delta \rceil}|$ is given by
the Catalan numbers, so $|\mathcal B_{k\lceil \log \Delta
\rceil}|=\frac{(2k\lceil \log \Delta \rceil)!}{(k\lceil \log \Delta
\rceil)!(k\lceil \log \Delta \rceil+1)!}=O^*(4^{k\log
\Delta})=O^*(\Delta^{2k})$. So, $\Sigma_{v \in V}|\mathcal
C_{k,v}|=O^*(\Delta^{2k})$. The proof of
Lemma~\ref{connectedinduced} is now completed.
\end{proof}
Consider now the following algorithm.
\begin{algo}[\texttt{ALG2}($T$,$k$)]
set $T = \emptyset$;\\
\texttt{ALG2}($T$,$k$)
\begin{itemize}
\item if $k > 0$ then, for each $i$ from $1$ to $k$,
\begin{itemize}
\item find $S_i \in V \setminus T$ minimizing $val(T \cup S_i)$
      with $S_i$ inducing a connected component of size $i$.
\item for each $i$, for each $v \in S_i$, run $\texttt{ALG2}(T \cup \{v\}$,$k-1)$;
\end{itemize}
\item else ($k=0$), stock the feasible solution $T$.
\end{itemize}
output the stocked feasible solution $T$ minimizing $val(T)$.
\end{algo}
The branching tree of \texttt{ALG2} 
has size $O(k^{2k})$.
Computing the $S_i$ in each
node takes time $O^*(\Delta^{2k})$ according to
Lemma~\ref{connectedinduced}. Thus, the algorithm runs in
$O^*((\Delta k)^{2k})$.

\begin{figure}
\label{fig2}
\begin{center}
\includegraphics[scale=.8]{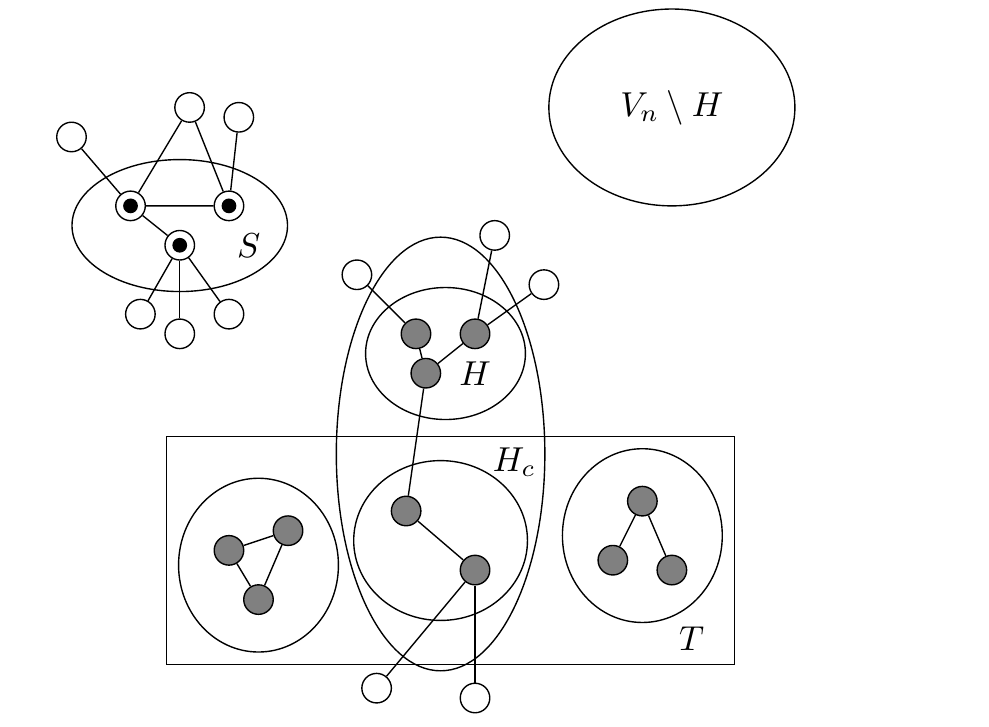}
\caption{Illustration of the proof, with filled
vertices representing the optimal solution
$V'_{\mathrm{opt}}$ and dotted vertices representing
the set $S=S_{|H|}$ computed by \texttt{ALG2}
which can substitute $H$, since $V_n$ does not
interact with $H_c$ nor with $S$.}
\end{center}
\end{figure}

For the optimality of \texttt{ALG2}, we use the following lemma.
\begin{lemma}\label{indep}
Let $A$,$B$,$X$,$Y$ be pairwise disjoint sets of vertices such that
$\val{}(A \cup X) \leqslant \val{}(B \cup X)$, $N[A] \cap Y =
\emptyset$ and $N[B] \cap Y = \emptyset$. Then, $\val{}(A \cup X
\cup Y) \leqslant \val{}(B \cup X \cup Y)$.
\end{lemma}
\begin{proof}[Proof of Lemma~\ref{indep}]
Simply observe that $\val{}(A \cup X \cup Y)=\val{}(Y)+\val{}(A \cup
X)-2\alpha_2|E(X,Y)|+\alpha_1|E(X,Y)| \leqslant \val{}(Y)+\val{}(B
\cup X)-2\alpha_2|E(X,Y)|+\alpha_1|E(X,Y)|=\val{}(B \cup X \cup Y)$,
that completes the proof of the lemma.
\end{proof}
We now show that \texttt{ALG2} is sound, using again hybridation
between an optimal solution $V'_{\mathrm{opt}}$ and the one solution
found by \texttt{ALG2}. We keep the same notation as in the proof of
the soundness of \texttt{ALG1}. Node~$s$ is a node of the branching
tree which deviates from $V'_{\mathrm{opt}}$, all nodes in the
branch between the root and $s$ are conform to $V'_{\mathrm{opt}}$,
the shared choices constitute the set of vertices $T=T(s)$ and, for
each~$i$, set $S_i=S_i(s)$ (analogously to $v(s)$ in the previous
proof, $s$ is now linked to the subsets $S_i$ computed at this
node). Set $V_n=V'_{\mathrm{opt}} \setminus T$. Take a maximal
connected (non empty) subset $H$ of $V_n$. Set $S=S_{|H|}$ and
consider $V_e= V'_{\mathrm{opt}} \setminus H \cup S = (T \cup V_n)
\setminus H \cup S = T \cup S \cup (V_n \setminus H)$. Note that, by
hypothesis, $N[S] \cap V_n = \emptyset$ since $s$ is a deviating
node. By the choice of $S$ at the node $s$, $\val{}(T \cup S)
\leqslant \val{}(T \cup H)$. So, $\val{}(V_e) = \val{}(T \cup S \cup
(V_n \setminus H)) =
     \val{}(T \cup H \cup (V_n \setminus H))
     = \val{}(T \cup V_n) = \val{}(V'_{\mathrm{opt}})$ according to Lemma~\ref{indep},
    since by construction neither~$N[H]$ nor~$N[S]$, do intersect $V_n \setminus H$.
Iterating the argument at most $k$ times we get to a leaf of the
branching tree of \texttt{ALG2} which yields a solution as good as
$V'_{\mathrm{opt}}$. The proof of the theorem is now completed.
\end{proof}
\begin{corollary}\label{cor11}
 \textsc{$k$-densest subgraph} and \textsc{min $(k,n-k)$-cut}
 can be solved in $O^*((\Delta k)^{2k})$.
\end{corollary}
Here also, simply observe that the problems mentioned in Corollary~\ref{cor11} are local graph partitioning problems.

Theorems~\ref{th1} and~\ref{th2} improve the $O^*(2^{(\Delta+1)k}$
$((\Delta+1)k)^{\log((\Delta+1)k)})$ time complexity for the
corresponding problems given in~\cite{Cai06randomseparation}
obtained there by the \textit{random separation} technique. Recall
that random separation consists of randomly guessing if a vertex is
in an optimal subset $V'$ of size $k$ (white vertices) or if it is
in $N(V') \setminus V'$ (black vertices). For all other vertices the
guess has no importance. As a right guess concerns at most only
$k+k\Delta$ vertices, it is done with high probability if we repeat
random guesses $f(k,\Delta)$ times with a suitable function~$f$.
Given a random guess, i.e., a random function $g:V \rightarrow
\{\text{white,black}\}$, a solution can be computed in polynomial
time by dynamic programming. Although random separation
(and {a fortiori} \textit{color coding}~\cite{Aloncolorcoding})  have also been applied to other problems than local graph partitioning ones, greediness-for-parameterization seems to be quite general and improves both
running time and easiness of implementation since our algorithms do
not need complex derandomizations.

Let us note that the greediness-for-parameterization technique can be even more general, by
enhancing the scope of Definition~\ref{local} and can be applied to
problems where the objective function takes into account not only
edges but also vertices. The {value} of a solution could be defined
as a function $\val{} : \mathcal P(V) \rightarrow \mathbb R$ such
that $\val{}(\emptyset)=0$, the contribution of a vertex $v$ in a
partial solution $T$ is $\delta(v,T)=\val{}(T \cup v) - \val{}(T)$.
Thus, for any subset $T$, $\val{}(T)=\val{}(T \setminus \{v_k\})+
\delta(v_k,T \setminus \{v_k\})$
 where $k$ is the size of $T$ and $v_k$ is the last vertex added to
the solution. Hence, $\val{}(T) =\Sigma_{1 \leqslant i \leqslant
k}\delta(v_i,\{v_1,\ldots,v_{i-1}\})+\val{}(\emptyset) =\Sigma_{1
\leqslant i \leqslant k}\delta(v_i,\{v_1,\ldots,v_{i-1}\})$. Now,
the only hypothesis we need to show Theorem~\ref{th2} is the
following: for each $T'$ such that $(N(T')\setminus T) \cap
(N(v)\setminus T)=\emptyset$, $\delta(v,T \cup T')=\delta(v,T)$.

Notice also that, that under such modification, \textsc{max $k$-dominating
set}, asking for a set~$V'$ of~$k$ vertices that dominate the
highest number of vertices in $V \setminus V'$ fulfils the
enhancement just discussed. We therefore derive the following.
\begin{corollary}
\textsc{max $k$-dominating set} can be solved in $O^*((\Delta k)^{2k}$.
\end{corollary}

\begin{figure}
\begin{center}
\includegraphics[scale=.7]{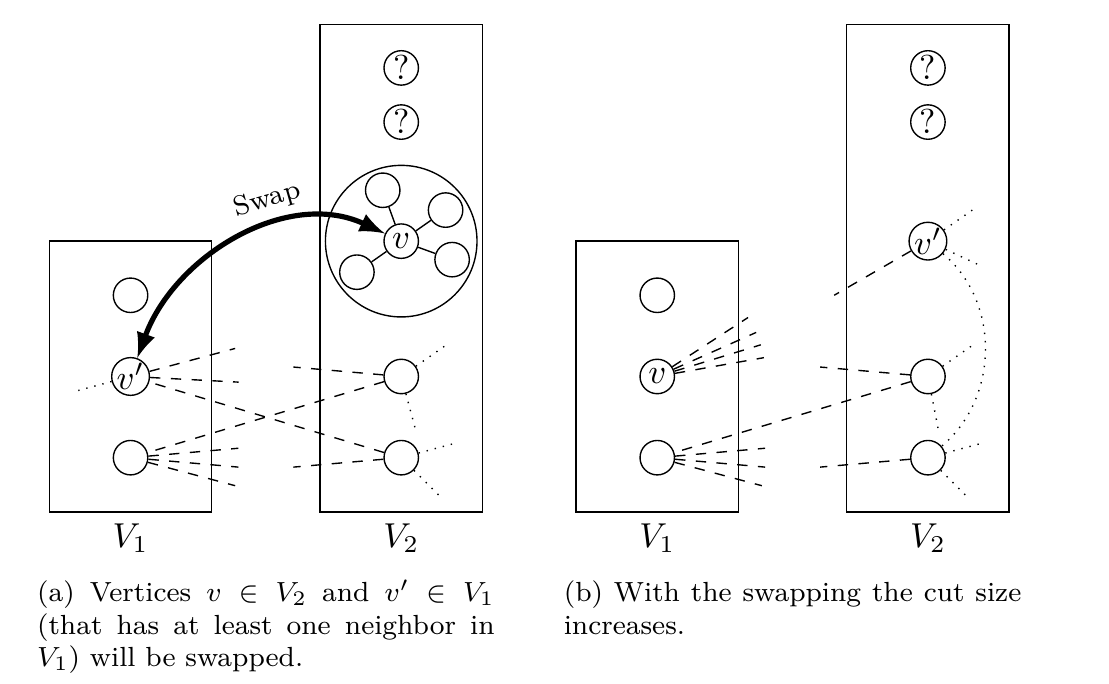}
\caption{Illustration of a swapping}
\label{exchange}
\end{center}
\end{figure}

\section{Standard parameterization for \textsc{max} and \textsc{min $\mathbf{(k,n-k)}$-cut}}\label{standardparam}

\subsection{\textsc{Max $\mathbf{(k,n-k)}$-cut}}

In the sequel, we use the standard notation $G[U]$ for any $U \subseteq V$
to denote the subgraph induced by the vertices of $U$.
In this section, we show that \textsc{max $(k,n-k)$-cut}
parameterized by the standard parameter, i.e., by the value~$p$ of
the solution, is FPT.
Using an idea of bounding above the value of an optimal solution by
a swapping process (see Figure~\ref{exchange}), we show that the non-trivial case satisfies $p > k$. We also show that $p > \Delta$
holds for non trivial instances and get the situation illustrated in~Figure~\ref{pkd}. The rest of the proof is an immediate application of Corollary~\ref{cor6}.
\begin{figure}
\begin{center}
%
%
\includegraphics[scale=.8]{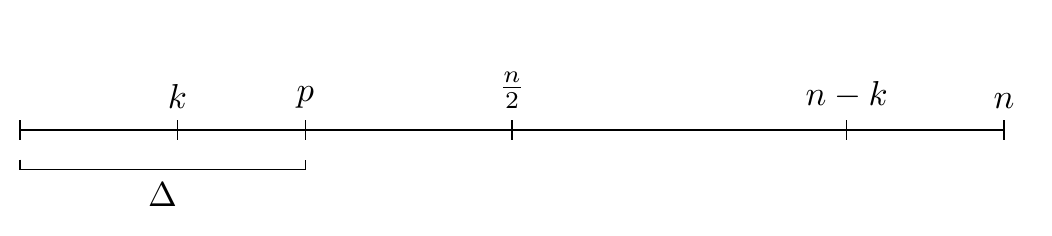}
\caption{Location of parameter $p$, relatively to $k$ and $\Delta$.}
\label{pkd}
\end{center}
\end{figure}
\begin{lemma}\label{reg}
In a graph with minimum degree $r$,
the optimal value $\mathrm{opt}$ of a \textsc{max} \knkcut{}
satisfies $\mathrm{opt} \geqslant \min\{n-k,rk\}$.
\end{lemma}
\begin{proof} 
We divide arbitrarily the vertices of a graph $G=(V,E)$
into two subsets~$V_1$ and~$V_2$ of size~$k$
and $n-k$, respectively. Then, for every vertex $v \in V_2$, we check if~$v$ has a
neighbor in~$V_1$. If not, we try to swap~$v$ and a vertex $v' \in V_1$ which has strictly less than~$r$ neighbors in~$V_2$ (see
Figure~\ref{exchange}). If there is no such vertex, then every vertex in~$V_1$ has at least~$r$ neighbors in~$V_2$, so determining a cut of value at least~$rk$. When swapping is possible,
 as the minimum degree is~$r$ and the neighborhood of~$v$ is entirely contained
in~$V_2$, moving~$v$ from~$V_2$ to~$V_1$ will increase the value of the cut by at least~$r$.
On the other hand, moving~$v'$ from~$V_1$ to~$V_2$ will reduce the value
 of the cut by at most $r-1$. In this way, the value of the cut
 increases by at least~1.

Finally, either the process has
reached a cut of value $rk$ (if no more swap is possible),  or every vertex in $V_2$ has increased the value of the cut by at least~1
(either immediately, or after a swapping process), which results in a cut of
value at least $n-k$, and the proof of the lemma is completed.
\end{proof}
\begin{corollary}\label{cormin}
In a graph with no isolated vertices, the optimal value for \textsc{max $(k,n-k)$-cut} is at least
$\min\{n-k,k\}$.
\end{corollary}
Then, Corollary~\ref{cor6} suffices to conclude the proof of the the following theorem.
\begin{theorem}\label{maxFPTp}
The \textsc{max $(k,n-k)$-cut} problem parameterized by the standard parameter~$p$ is FPT.
\end{theorem}

\subsection{\textsc{Min $\mathbf{(k,n-k)}$-cut}}

Unfortunately, unlike what have been done for
\textsc{max $(k,n-k)$-cut}, we have not been able to show until now
that the case $p<k$ is ``trivial''. So, Algorithm~\texttt{ALG2} in Section~\ref{greedy-param} cannot be transformed into a standard FPT algorithm for this problem.

However, we can prove that when
$p\geqslant k$, then \textsc{min $(k,n-k)$-cut} parameterized by the
value $p$ of the solution is FPT. This is an immediate corollary of
the following proposition.
\begin{proposition}\label{minFPTp}
\textsc{min $(k,n-k)$-cut} parameterized by $p+k$ is FPT.
\end{proposition}
\begin{proof}
Each vertex $v$ such that $|N(v)| \geqslant k+p$ has to be in $V
\setminus V'$ (of size $n-k$). Indeed, if one puts~$v$ in~$V'$ (of
size $k$), among its $k+p$ incident edges, at least $p+1$ leave from
$V'$; so, it cannot yield a feasible solution. All the vertices $v$
such that $|N(v)| \geqslant k+p$ are then rejected. Thus, one can
adapt the FPT algorithm in $k+\Delta$ of Theorem~\ref{th2} by
considering the $k$-neighborhood of a vertex $v$ not in the whole
graph~$G$, but in~$G[T \cup U]$. One can easily check that the
algorithm still works and since in those subgraphs the degree is
bounded by $p+k$ we get an FPT algorithm in $p+k$.
\end{proof}
In~\cite{Feigeminkcut}, it is shown that, for any $\varepsilon > 0$,
there exists a randomized $(1+\frac{\varepsilon k} {\log
n})$-approximation for \textsc{min $(k,n-k)$-cut}. From this result,
we can easily derive that when $p < \frac{\log n}{k}$  then the
problem is solvable in polynomial time (by a randomized algorithm).
Indeed, fixing $\varepsilon=1$, the algorithm in~\cite{Feigeminkcut}
is a $(1+\frac{k}{\log(n)})$-approximation. This approximation ratio
is strictly better than $1+\frac{1}{p}$. This means that the
algorithm outputs a solution of value lower than $p+1$, hence at
most $p$, if there exists a solution of value at most $p$.

We now conclude this section by claiming that, when $p \leqslant k$, \textsc{min $(k,n-k)$-cut} can be solved in time~$O^*(n^p)$.
\begin{proposition}\label{n^p}
If $p \leqslant k$, then \textsc{min $(k,n-k)$-cut} can be solved in time~$O^*(n^p)$.
\end{proposition}
\begin{proof}
Since $p \leqslant k$, there exist in the optimal set~$V'$, $p' \leqslant p$ vertices incident to the~$p$ outgoing edges. So,  the $k-p'$ remaining vertices of~$V'$ induce a subgraph that is
disconnected from~$G[V \setminus V']$.

Hence, one can enumerate all the $p' \leqslant p$ subsets of~$V$. For
each such subset~$\widetilde{V}$, the graph~$G[V \setminus
\widetilde{V}]$ is disconnected. Denote by $C=(C_i)_{0\leqslant i
\leqslant |C|}$ the connected components of~$G[V \setminus
\widetilde{V}]$ and by~$\alpha_i$ the number of edges between $C_i$ and
$\widetilde{V}$. We have to pick a subset $C' \subset C$ among these
components such that $\sum_{C_i \in C'} |C_i| = k-p'$ and
maximizing~$\sum_{C_i\in C'}\alpha_i$. This can be done in polynomial time using standard
dynamic programming techniques.
\end{proof}

\section{Other parameterizations}

\subsection{Parameterization by~$k$ and approximation of \textsc{max} and \textsc{min} \textsc{$\mathbf{(k,n-k)}$-cut}}\label{paramapprox}

Recall that both \textsc{max} and \textsc{min $(k,n-k)$-cut} parameterized by $k$
are W[1]-hard~\cite{Downey03cuttingup,cai}. In this section, we give some approximation algorithms working in FPT time with respect to parameter~$k$. 
\begin{proposition}\label{k-approx-max&min}
\textsc{max $(k,n-k)$-cut}, parameterized by~$k$ has a fixed-parameter approximation schema. On the other hand, \textsc{min $(k,n-k)$-cut} parameterized by~$k$ has a randomized fixed-parameter approximation schema.
\end{proposition}
\begin{proof}
We first handle \textsc{max} \knkcut{}. Fix some $\varepsilon > 0$. Given a graph $G=(V,E)$, let $d_1
\leqslant d_2 \leqslant \ldots \leqslant d_k$ be the degrees of the
$k$ largest-degree vertices $v_1, v_2, \ldots v_k$ in $G$. An optimal
solution of value~opt is obviously bounded from above by $B
=\Sigma_{i=1}^k d_i$. Now, consider solution
$V'=\{v_1,v_2,\ldots,v_k\}$. As there exist at most $k(k-1)/2
\leqslant k^2/2$ (when $V'$ is a $k$-clique) inner edges,
solution~$V'$ has a value~sol at least $B-k^2$. Hence, the
approximation ratio is at least $\frac{B-k^2}{B} = 1 - \frac
{k^2}{B}$. Since, obviously, $B \geqslant d_1 = \Delta$, an approximation ratio at least $1 - \frac{k^2}{\Delta}$ is immediately derived.

If $\varepsilon \geqslant \frac{k^2}{\Delta}$ then $V'$ is a $(1-\varepsilon)$-approximation. Otherwise, if $\varepsilon \leqslant \frac{k^2}{\Delta}$, then $\Delta \leqslant \frac{k^2}{\varepsilon}$. So, the branching algorithm of Theorem~\ref{maxFPTp} with time-complexity~$O^*(\Delta^k)$ is in this case an $O^*(\frac{k^{2k}}{\varepsilon^k})$-time algorithm.

For \textsc{min} \knkcut{}, it is proved in~\cite{Feigeminkcut} that, for  $\varepsilon > 0$, if $k<\log n$, then there exists a randomized polynomial time $(1+\varepsilon)$-approximation. Else, if $k>\log n$, the exhaustive enumeration of the $k$-subsets takes time $O^*(n^k)=O^*((2^k)^k)=O^*(2^{k^2})$.
\end{proof}
Finding approximation algorithms that work in FPT time with respect
to parameter~$p$ is an interesting question. Combining  the result
of~\cite{Feigeminkcut} and an $O(\log^{1.5}(n))$-approximation
algorithm in~\cite{feige} we can show that the problem is~$O(k^{3/5})$ approximable in polynomial time by a randomized
algorithm. But, is it possible to improve this ratio when allowing FPT
time (with respect to $p$)?

\subsection{Parameterization by the treewidth and the vertex co\-ver number}\label{otherparam}

When dealing with parameterization of graph problems, some classical
parameters arise naturally. One of them, very frequently used in the
fixed parameter literature is the treewidth of the graph.

It has already been proved that \textsc{min} and \textsc{max $(k,n-k)$-cut}, as well as \textsc{$k$-densest subgraph} can be solved in $O^*(2^{\mathrm{tw}})$~\cite{kdensestwalcom,Kloks94}. We show here that the algorithm in~\cite{kdensestwalcom} can be adapted to handle the whole class of local problems, deriving so the following result.
\begin{proposition}\label{tw}
Any local graph partitioning problem can be solved in time $O^*(2^{\mathrm{tw}})$.
\end{proposition}
\begin{proof}
A tree decomposition of a graph~$G(V,E)$ is a pair~$(X,T)$ where~$T$
is a tree on vertex set~$N(T)$ the vertices of which are called
nodes and $X = (\{X_i: i \in N(T)\})$ is a collection of subsets
of~$V$ such that: (i)~$\cup_{i \in N(T)}X_i = V$, (ii)~for each edge
$(v,w) \in E$, there exist an $i \in N(T)$ such that $\{v, w\} \in
X_i$, and (iii)~for each $v \in V$, the set of nodes $\{i: v \in
X_i\}$ forms a subtree of~$T$. The width of a tree decomposition
$(\{X_i: i \in N(T)\},T)$ equals $\max_{i \in N(T)}\{|X_i| - 1\}$.
The treewidth of a graph~$G$ is the minimum width over all tree
decompositions of~$G$. We say that a tree decomposition is nice if
any node of its tree that is not the root is one of the following
types:
\begin{itemize}
\item a leaf that contains a single vertex from the graph;
\item an introduce node $X_i$ with one child $X_j$ such that $X_i=X_j \cup \{v\}$ for some vertex $v \in V$;
\item a forget node $X_i$ with one child $X_j$ such that $X_j=X_i \cup \{v\}$ for some vertex $v \in V$;
\item a join node $X_i$ with two children $X_j$ and $X_l$ such that $X_i=X_j=X_l$.
\end{itemize}
Assume that the local graph partitioning problem~$\Pi$ is a
minimization problem (we want to find $V'$ such that
$\mathrm{val}(V') \leqslant p$), the maximization case being
similar. An algorithm that transforms in linear time an arbitrary
tree decomposition into a nice one with the same treewidth is
presented in~\cite{Kloks94}. Consider a nice tree decomposition of
$G$ and let $T_i$ be the subtree of $T$ rooted at $X_i$, and
$G_i=(V_i,E_i)$ be the subgraph of $G$ induced by the vertices in
$\bigcup_{X_j \in T_i} X_j$. For each node
$X_i=(v_1,v_2,\ldots,v_{|X_i|})$ of the tree decomposition, define a
configuration vector $\vec{c} \in \{0,1\}^{|X_i|}$; $\vec{c}[j]=1
\Longleftrightarrow v_j \in X_i$ belongs to the solution. Moreover,
for each node~$X_i$, consider a table $A_i$ of size $2^{|X_i|}
\times (k+1)$. Each row of $A_i$ represents a configuration and each
column represents the number~$k'$, $0 \leq k' \leq k$, of vertices
in $V_i \setminus X_i$ included in the solution. The value of an
entry of this table equals  the value of the best solution
respecting both the configuration vector and the number $k'$, and
$-\infty$ is used to define an infeasible solution. In the sequel,
we set $X_{i,t}=\{v_h \in X_i: \vec{c}(h)=1\}$
and $X_{i,r}=\{v_h \in X_i: \vec{c}(h)=0\}$.

The algorithm examines the nodes of $T$ in a bottom-up way and
fills in the table $A_i$ for each node~$X_i$.
In the initialization step, for each leaf node~$X_i$ and each
configuration~$\vec{c}$,
we have $A_i[\vec{c},k']=0$ if $k'=0$; otherwise $A_i[\vec{c},k']=-\infty$.

If $X_i$ is a forget node, then consider a configuration~$\vec{c}$
for~$X_i$. In $X_j$ this configuration is extended with the decision
whether vertex $v$ is included into the solution or not. Hence, taking
into account that $v \in V_i \setminus X_i$ we get:
$$
A_i\left[\vec{c},k'\right] = \min\left\{A_j\left[\vec{c} \times
\{0\},k'\right],A_j\left[\vec{c} \times \{1\},k'-1\right]\right\}
$$
for each configuration $\vec{c}$ and each $k'$, $0 \leq k' \leq k$.

If $X_i$ is an introduce node, then consider a configuration
$\vec{c}$ for $X_j$. If $v$ is taken in $V'$,
its inclusion adds the quantity
$\delta_v=\alpha_1|E(\{v\},X_{i,t})|+\alpha_2|E(\{v\},X_{i,r})|$
to the solution. The crucial point is
that~$\delta_v$ does not depend on the $k'$ vertices
of $V_i \setminus X_i$ taken in the solution. Indeed, by
construction a vertex in $V_i \setminus X_i$ has its subtree
entirely contained in $T_i$. Besides, the subtree of $v$ intersects~$T_i$
only in its root, since $v$ appears in $X_i$, disappears from~$X_j$ and has,
by definition, a connected subtree.
So, we know that there is no edge in $G$ between
$v$ and any vertex of $V_i \setminus X_i$. Hence, $A_i[\vec{c}
\times \{1\},k'] = A_j[\vec{c},k'] + \delta_v$, since $k'$ counts
only the vertices of the current solution in $V_i \setminus X_i$.
The case where $v$ is discarded from the solution (not taken in
$V'$) is completely similar; we just define $\delta_v$ according to
the number of edges linking $v$ to vertices of $T_i$ respectively in
$V'$ and not in~$V'$.

If $X_i$ is a join node, then for each configuration~$\vec{c}$ for~$X_i$ and each $k'$, $0 \leq k' \leq k$, we have to find the best
solution obtained by $k_j$, $0 \leq k_j \leq k'$, vertices in~$A_j$
plus $k'-k_j$ vertices in~$A_l$.
However, the quantity $\delta_{\vec{c}} =
\alpha_1|E(X_{i,t})|+\alpha_2|E(X_{i,t},X_{i,r})|$
is counted twice. Note that~$\delta_{\vec{c}}$ depends only on~$X_{i,t}$ and~$X_{i,r}$,
since there is no edge between $V_l \setminus X_i$ and $V_j \setminus X_i$.
Hence, we get:
$$
A_i\left[\vec{c},k'\right] = \max_{0 \leq k_j \leq
k'}\left\{A_j\left[\vec{c},k_j\right]+A_l\left[\vec{c},k'-k_j\right]\right\}
- \delta_c
$$
and the proof of the proposition is completed.
\end{proof}
\begin{corollary}
Restricted to trees, any local graph partitioning problem can be solved in polynomial time.
\end{corollary}
\begin{corollary}
\textsc{min bisection} parameterized by the treewidth of the input graph is FPT.
\end{corollary}
It is worth noticing that the result easily extends to the weighted
case (where edges are weighted) and to the case of partitioning $V$ into a
constant number of classes (with a higher running time).

Another natural parameter frequently used in the parameterized complexity framework is the size~$\tau$ of a minimum vertex cover of the input graph. Since it always holds that $\mathrm{tw} \leqslant \tau$, the result of Proposition~\ref{tw} immediately applies to parameterization by~$\tau$. However, the algorithm developed there needs exponential space. In what follows, we give a simple parameterization by~$\tau$ using polynomial space.
\begin{proposition}
\label{FPTtau}
\textsc{max} and \textsc{min $(k,n-k)$-cut} parameterized by~$\tau$ can be solved in~FPT~$O^*(2^\tau)$ time and in polynomial space.
\end{proposition}
\begin{proof}
Consider the following algorithm:
\begin{itemize}
\item compute a minimum vertex cover~$C$ of~$G$;
\item for every subset~$X$ of~$C$ of size~$|X|$ smaller than~$k$, complete~$X$ with the $k-|X|$ vertices
of $V \setminus C$ that maximize (resp., minimize) their incidence
with~$C\setminus X$ (i.e., the number of neighbours in~$C\setminus
X$);
\item output the best solution.
\end{itemize}
Recall that a minimum size vertex cover can be computed in time~$O^*(1.2738^{\tau})$ time by means of the fixed-parameter algorithm of~\cite{chenetal_vc_tcs10} and using polynomial space. The operation on every subset is polynomial, so the global computation time is at most~$O^*(2^\tau)$.

The soundness follows from the fact that a complement of a vertex cover is an independent set. Denoting by~$V'$ the optimal vertex-set (i.e., the~$k$ vertices inducing an optimal cut), then $V' \cap C$ will be considered by the above algorithm, and then every vertex of the completion will add exactly to the solution its number of neighbors in $V' \cap C$, which is maximized (or minimized) in the algorithm.
\end{proof}

\end{document}